\documentclass[a4paper,UKenglish,twocolumn]{article}
\pdfoutput=1

\usepackage{scrextend}
\changefontsizes[10pt]{8pt}

\usepackage{lmodern}
\usepackage[sc]{mathpazo} 
\usepackage[euler-digits,euler-hat-accent]{eulervm} 
\usepackage[T1]{fontenc} 
\usepackage[utf8]{inputenc}
\usepackage{microtype} 

\usepackage[hmarginratio=1:1,top=20mm,left=22mm,columnsep=20pt]{geometry} 
\usepackage[hang, small,labelfont=bf,up,textfont=it,up]{caption} 
\usepackage{nameref}
\usepackage{hyperref}
\usepackage{cite}
\usepackage[capitalise]{cleveref}

\usepackage{lettrine} 
\usepackage{paralist} 

\setdefaultleftmargin{1.5em}{1em}{}{}{}{}

\usepackage{titlesec}
\titleformat{\section}[block]{\large\scshape\centering}{\thesection.}{1em}{} 
\titleformat{\subsection}[block]{\large\scshape\centering}{\thesubsection.}{1em}{} 

\usepackage{fancyhdr} 
\usepackage{amsmath,amsthm,amssymb,centernot,mathtools,thmtools}

\usepackage{tikz}
\usetikzlibrary{calc,arrows,matrix,positioning,intersections}

\tikzset{nat/.style={double,double equal sign distance}}
\tikzset{nat>/.style={nat,-implies}}
\tikzset{<nat/.style={nat,implies-}}
\tikzset{descr/.style={anchor=center,fill=white}}

\usepackage{babel}
\usepackage{xpunctuate}
\newcommand{\ie}{\emph{i.e.}\xspace}
\newcommand{\eg}{\emph{e.g.}\xspace}
\newcommand{\cf}{\emph{cf.}\xspace}
\newcommand{\etc}{\xperiodafter{\emph{etc}}}
\newcommand{\loccit}{\emph{loc.~cit.}\xspace}
\newcommand{\adhoc}{\emph{ad hoc}\xspace}
\newcommand{\viceversa}{\emph{vice versa}\xspace}
\newcommand{\mutatismutandis}{\emph{mutatis mutandis}\xspace}
\newcommand{\criterium}{\emph{criterium}\xspace}

\theoremstyle{plain}
\newtheorem{theorem}{Theorem}[section]
\newtheorem{corollary}[theorem]{Corollary}
\newtheorem{proposition}[theorem]{Proposition}
\newtheorem{lemma}[theorem]{Lemma}
\theoremstyle{definition}
\newtheorem{definition}[theorem]{Definition}
\theoremstyle{remark}
\newtheorem{remark}[theorem]{Remark}
\newtheorem{example}[theorem]{Example}

\makeatletter
\newcommand{\superimpose}[2]{%
  {\ooalign{$#1\@firstoftwo#2$\cr\hfil$#1\@secondoftwo#2$\hfil\cr}}}
\makeatother

\makeatletter
	\newcommand{\footnoteref}[1]{%
		\protected@xdef\@thefnmark{\ref{#1}}\@footnotemark%
	}
\makeatother

\makeatletter
\newcommand{\customlabel}[2]{%
	\protected@write \@auxout {}{%
		\string \newlabel {#1}{{#2}{\thepage}{#2}{#1}{}}%
	}\hypertarget{#1}{#2}
}
\makeatother

\newcommand{\defeq}{\stackrel{\vartriangle}{=}}

\DeclareMathOperator{\img}{img}
\DeclareMathOperator{\obj}{obj}

\newcommand{\bottom}{\perp}

\newcommand{\bef}[1]{\mathfrak{#1}}

\newcommand{\cat}[1]{\ensuremath{\mathbf{#1}}\xspace}

\newcommand{\VCat}[1][\cat{V}]{\textnormal{#1-\cat{Cat}}}
\newcommand{\Set}{\cat{Set}}
\newcommand{\Pos}{\cat{Pos}}
\newcommand{\Cpo}{\cat{Cpo}}
\newcommand{\Cpob}{\cat{Cpo_{\mspace{-2mu}\bottom\mspace{-2mu}}}}
\newcommand{\Kl}{\mathcal{K}\kern-\nulldelimiterspace{l}}
\newcommand{\Ch}{\mathcal{C}\kern-\nulldelimiterspace{h}}
\newcommand{\LCh}{\overline{\mathcal{C}\kern-\nulldelimiterspace{h}}}

\newcommand{\FCat}[2]{\ensuremath{[#1,#2]}}
\newcommand{\Endo}[1]{\FCat{#1}{#1}}
\newcommand{\EndoFact}[2]{\ensuremath{[#1,#2,#1]}}
\newcommand{\underlying}[1]{\lfloor#1\rfloor}
\newcommand{\lifted}[1]{\ensuremath{\cat{Lift}_{#1}}}
\newcommand{\dropped}[1]{\ensuremath{\cat{Drop}_{#1}}}

\newcommand{\Id}{I\mspace{-2mu}d}
\newcommand{\id}{i\mspace{-2mu}d}
\newcommand{\inj}{i\mspace{-2mu}n}

\newcommand{\op}{^{o\mspace{-2mu}p}}

\title{\fontsize{18pt}{10pt}\selectfont\textbf{\sc
	Endofunctors Modelling Higher-Order Behaviours
}}
\author{
\large\textsc{Marco Peressotti}\\[2pt]
\normalsize Laboratory of Models and Applications of Distributed Systems,\\[-2pt]
\normalsize	Department of Mathematics, Computer Science and Physics\\[-2pt]
\normalsize University of Udine, Italy\\
\normalsize \href{marco.peressotti@uniud.it}{\tt marco.peressotti@uniud.it}
\vspace{-2mm}
}
\date{}

\hypersetup{
	colorlinks=true,
	linkcolor=black,
	citecolor=black,
	filecolor=black,
	urlcolor=black,
	pdftitle={Endofunctors modelling higher-order behaviours},
	pdfauthor={Marco Peressotti}
}

\begin{document}

\maketitle

\thispagestyle{fancy} 


\begin{abstract}
	In this paper we show how the abstract behaviours of higher-order systems can be modelled as \emph{final coalgebras} of suitable \emph{behavioural functors}. These functors have the challenging peculiarity to be circularly defined with their own final coalgebra. Our main contribution is a general construction for defining these functors, solving this circularity which is the essence of higher-order behaviours. This characterisation is \emph{syntax agnostic}. To achieve this property, we shift from term passing to \emph{behaviour passing}: in the former higher-order is expressed by passing around syntactic objects--such as terms or processes--as representations of behaviours whereas the former ditches the syntactic encoding altogether and works directly with behaviours \ie semantic objects. From this perspective, the former can be seen as syntactic higher-order whereas the latter as \emph{semantic higher-order}.
\end{abstract}


\section{Introduction}
\label{sec:introduction}

It is well known that \emph{higher-order systems}, \ie systems which can pass
around systems of the same kind, like the $\lambda$-calculus \cite{baren:lcalc,
abramsky:rtfp1990}, the calculus of higher-order communicating systems (CHOCS)
\cite{thomsen:acta1993}, the higher-order $\pi$-calculus (HO$\pi$) 
\cite{sangiorgi:ic1996}, HOcore \cite{lpss:lics2008}, \etc, are
difficult to reason about. Many bisimulations and proof methods have been
proposed also in recent works \cite{sks:lics2007,sl:sas2009,
kls:mfps2011,lsjb:ic2011,bmss:lmcs2012,sumii:lics12}. This effort points out
that a definition of \emph{abstract} higher-order behaviour is still elusive.
In this paper, we show how these abstract behaviours can be modelled as the
\emph{final coalgebras} of suitable \emph{higher-order behavioural functors}.

\looseness=-1
Coalgebras are a well established framework for modelling and studying concurrent and reactive systems \cite{rutten:tcs2000}.  In this approach, we first define a \emph{behavioural endofunctor} $B$ over $\Set$ (or other suitable category), modelling the computational aspect under scrutiny; for $X$ a set of states, $BX$ is the type of behaviours over $X$.  Then, a system over $X$ corresponds to a $B$-coalgebra, \ie a map $h\colon X\to BX$ associating each state with its behaviour.  The crucial step of this approach is defining the functor $B$, as it corresponds to specify the behaviours that the systems are meant to exhibit.  Once we have defined a behavioural functor, many important properties and general results can be readily instantiated, such as the existence of the \emph{final $B$-coalgebra} (containing all abstract behaviours), the definition of the canonical \emph{coalgebraic bisimulation} (which is the abstract generalization of Milner's strong bisimilarity) and its coincidence with behavioural equivalence \cite{am:ctcs89}, the construction of canonical \emph{trace semantics} \cite{hjs:lmcs2007} and \emph{weak bisimulations} \cite{bmp:jlamp2015}, the notion of \emph{abstract GSOS} \cite{tp:lics1997,klin:tcs2011}, \etc.  We stress the fact that behavioural functors are ``syntax agnostic'': they define the semantic behaviours, abstracting from any specific concrete representation of systems. In the wake of these important results, many functors have been defined for modelling a wide range of behaviours: deterministic and non-deterministic systems \cite{rutten:tcs2000}; systems with I/O, with names, with resources \cite{ft:lics2001,fs:ic2006}; systems with quantitative aspects such as probabilities \cite{sokolova:tcs2011,doberkat:scl,ks:ic2013,mp:tcs2016,mp:qapl2014}; systems with continuous states \cite{bm:jcss2015}, \etc.

Despite these results, a general coalgebraic treatment of higher-order systems
is still missing.  In fact, defining endofunctors for higher-order behaviours
is challenging.  In order to describe the problem, let us consider first an endofunctor over \Set for representing the behaviour of a first-order calculus, like CCS with
value passing \cite{hi:ic1993}:
\begin{equation}
   B = \mathcal{P}_{\!\!\omega}(C\times V\times \Id + C\times \Id^V + \Id)
   \label{eq:eg-firstorder-beh}
\end{equation}
where $C$ is a set of channels and $V$ is the set of values
\cite{ft:lics2001}. This functor is well-defined, and it admits a final
coalgebra which we denote by $\nu B$; the carrier of this coalgebra is the set
$|\nu B|$ containing \emph{all possible (abstract) behaviours}, 
\ie,  synchronization trees labelled with nothing ($\tau$-actions), 
input or output actions.

In a higher-order calculus like HO$\pi$, the values that processes can communicate are processes themselves. However, actions communicating semantically equivalent (herein, strongly bisimilar) processes have to be considered equivalent even if the values/processes exchanged are syntactically different. In other words, this means that from the semantics perspective higher-order behaviours communicate \emph{behaviours}. To reflect this fundamental observation in the definition \eqref{eq:eg-firstorder-beh} we must replace the set of exchanged values $V$ with the set of all possible behaviours \ie the carrier of the final $B$-coalgebra yielding the following definition:
\begin{equation}
   B_{ho} = \mathcal{P}_{\!\!\omega}(C\times |\nu B_{ho}| \times \Id + 
   C\times \Id^{|\nu B_{ho}|} + \Id)
   \label{eq:eg-higherorder-beh}
\end{equation}
But this means that we are defining $B_{ho}$ using its own final coalgebra $\nu B_{ho}$, which can be defined (if it exists) only after $B_{ho}$ is defined---a circularity!

\looseness=-1
We think that this circularity is the gist of higher-order behaviours: any attempt to escape it would be restricting and distorting.  One may be tempted to take as $V$ some (syntactic) representation of behaviours (\eg, processes), but this would fall short.  First, the resulting behaviours would not be really higher-order, but rather behaviours manipulating some \adhoc representation of behaviours. Secondly, we would need some mechanism for moving between behaviours and their representations--which would hardly be complete. Third, the resulting functor would not be abstract and independent from the syntax of processes, thus hindering the possibility of reasoning about the computational aspect on its own, and comparing different models sharing the same kind of behaviour.

This fundamental shift from \emph{term/process passing} to \emph{behaviour passing} is at the hearth of this work and introduces a distinction between syntactic and semantic higher-order. Intuitively, for fully-abstract calculi this means that instead of restricting to systems unable to distinguish between bisimilar processes, we exchange collections of bisimilar processes/terms. This frees us from the hurdle of finding complete syntactic representations and canonical representatives. To some extent, this approach follows the idea of Sangiorgi's \emph{environmental bisimulation} \cite{sks:lics2007} where bisimulation relations are indexed by approximations of the bisimilarity relation used in place of syntactic equivalence during the bisimulation game.

The main contribution of this paper is a general characterisation endofunctors modelling higher-order operational behaviours in terms of solutions to certain recursive equations. We complement this result providing a categorical construction for finding solutions to these equations while dealing with the unavoidable circularity mentioned above. The key idea is to consider the definition as an instance of an endofunctor $\bef{F}(V)\colon \Set\to\Set$ parameterised in the object of values $V$. Then, we are interested in those instances whose final coalgebra is carried by the object of values \ie those such that values are exactly all (abstract) behaviours. Actually, since this parameter may occur in both covariant and contravariant position, the functor is \emph{biparametric}.  In our example, $\bef{F}\colon\Set^{op}\times\Set \to \Endo{\Set}$ is given as:
\begin{equation}
   \bef{F}(X,Y) = \mathcal{P}_{\!\!\omega}(C\times Y \times \Id + C\times \Id^X + \Id)
   \label{eq:eg-parametric-beh}
\end{equation}
where $\Endo{\Set}$ denotes the category of endofunctors over \Set. In general, we consider functors $\bef{F}\colon \cat{C}^{op}\times\cat{C} \to \Endo{\cat{C}}$ (eventually $\bef{F}\colon \cat{D}^{op}\times\cat{D} \to \Endo{\cat{C}}$) and show how to define an initial sequence of endofunctors together with their final coalgebras such that its limit $(B, |\nu B|)$ exists and satisfies:
\[
	B \cong \bef{F}(|\nu B|,|\nu B|).
\]
Thus, $B$ is the requested higher-order behavioural functor.

A consequence of this construction is that we can now apply standard results and techniques offered by the coalgebraic framework and \eg, derive a canonical higher-order bisimulation as an instance of Aczel-Mendler's coalgebraic bisimulation. Likewise, we can derive SOS specifications by instantiating the bialgebraic framework thus accommodating syntactic representations (\eg processes) in our settings.

The remainder of this paper is structured as follows: \cref{sec:algebraic-compactness} contains preliminaries on order-enriched categories and algebraic compactness. \Cref{sec:hob} provides an abstract characterisation of endofunctors modelling higher-order behaviours and a construction for computing them levering algebraic compactness. \Cref{sec:hob-generalised} refines and generalises this construction with the relevant (albeit technical) benefit of reducing the parts where algebraic compactness is assumed. \Cref{sec:conclusions} contains some concluding remarks and directions for further work.

\section{Preliminaries on algebraic compactness}
\label{sec:algebraic-compactness}

In this section we recall preliminary notions and results relevant relevant to the constructions described in the following sections. In particular, we need a form of \emph{limit-colimit coincidence} \cite{scott:lnm1972} developed in the field of (categorical) domain theory in order to guarantee the existence of (unique dominating) solutions to equations with unknowns occurring in both covariant and contravariant position.

Let \Cpo be the category whose objects are (small) $\omega$-complete partial orders and whose morphisms are continuous maps and let \Cpob be its subcategory whose objects have bottoms and whose morphisms are bottom-strict.

A \emph{\Cpo-enriched category} (or simply \Cpo-category) \cat{C} is a locally small category whose hom-sets $\cat{C}(X,Y)$ come equipped with an $\omega$-complete partial order $\leq_{X,Y}$ such that composition $(-\circ-)\colon \cat{C}(Y,Z)\times\cat{C}(X,Y) \to \cat{C}(X,Z)$ is a continuous operation. A special case of \Cpo-categories are those enriched over \Cpob \ie any \Cpo-category \cat{C} whose hom-sets $\cat{C}(X,Y)$ are additionally equipped with a bottom element $\bottom_{X,Y}$ and whose composition operation is strict. We shall drop subscripts from $\leq_{X,Y}$ and $\bottom_{X,Y}$ when possible. We denote the category underlying a \Cpo-category \cat{C} as $\lfloor\cat{C}\rfloor$ or just as \cat{C} when clear from the context.

In the following let \cat{V} stand for either \Cpo or \Cpob.

\begin{example}
	The category \cat{V} is enriched over itself. The single object category \cat{1} is trivially \Cpob-enriched. The dual of a \cat{V}-category \cat{C} is the \cat{V}-category $\cat{C}\op$ such that $\obj(\cat{C}\op) = \obj(\cat{C})$ and $\cat{C}\op(X,Y) = \cat{C}(Y,X)$. The product of \cat{V}-categories \cat{C} and \cat{D} is the \cat{V}-category $\cat{C}\times\cat{D}$ such that $(\cat{C}\times\cat{D})((X,X'),(Y,Y')) = \cat{C}(X,Y)\times\cat{D}(X',Y')$. The category of relations $\cat{Rel} \cong \Kl(\mathcal{P})$ is a \Cpob-category where the order structure is defined by pointwise extension of the inclusion order created by the powerset monad---see \eg \cite{hjs:lmcs2007,bmp:jlamp2015} for more behavioural functors of endowed with monadic structures yielding \cat{V}-enriched Kleisli categories.
\end{example}

A \emph{\cat{V}-enriched functor} $F\colon \cat{C} \to \cat{D}$ between \cat{V}-categories is a functorial mapping such that for every pair of objects $X$, $Y$ in $\cat{C}$ the assignment $F_{X,Y}\colon \cat{C}(X,Y) \to \cat{D}(FX,FY)$ is continuous and, in the case of \Cpob-functors, strict. Unless otherwise stated, we implicitly assume that domain and codomain of an enriched functor are similarly enriched. We shall denote the functor underlying a \cat{V}-functor $F\colon \cat{C} \to \cat{D}$ as $\lfloor F \rfloor\colon \lfloor\cat{C}\rfloor \to \lfloor\cat{D}\rfloor$ or simply $F$, when confusion seems unlikely. For \cat{V}-categories \cat{C} and \cat{D}, the functor category
$[\cat{C},\cat{D}]$ is the \cat{V}-category whose objects are \cat{V}-functors and such that $[\cat{C},\cat{D}](X,Y)$ is the complete partial order on the set $\cat{Nat}(X,Y)$ of natural transformations given by pointwise extension of the order on their components. Categories and functors enriched over \cat{V} form the category \VCat.
A \emph{\Cpo-adjunction} $\chi\colon L\vdash R\colon \cat{C} \to \cat{D}$ is given by a natural isomorphism 
\[
	\chi\colon \cat{C}(L-,-) \cong \cat{D}(-,R-)\colon \cat{D}\op\times\cat{C} \to \Pos
	\text{.}
\]
Actually, the above statement defines a \Pos-adjunction but, since the inclusion functor $\Cpo \to \Pos$ creates isomorphisms, any \Pos-adjunction involving \Cpo-categories yields a \Cpo-adjunction.

Two morphisms $e\colon X \to Y$ and $p\colon Y \to X$ in a \Cpo-category \cat{C} form an \emph{embedding-projection pair} (written $e \lhd p \colon X \to Y$) whenever $p\circ e = \id_X$ and $e\circ p \leq \id_Y$ or, diagrammatically:
\[\begin{tikzpicture}[font=\footnotesize,auto,
	xscale=1, yscale=1,
	baseline=(current bounding box.center),
	]
	\node (n0) at (0,1) {\(X\)};
	\node (n1) at (1,1) {\(Y\)};
	\node (n2) at (1,0) {\(X\)};
	\node (n3) at (2,0) {\(Y\)}; 

	\draw[>->] (n0) to node {\(e\)} (n1);
	\draw[->>] (n1) to node[swap] {\(p\)} (n2);
	\draw[>->] (n2) to node[swap] {\(e\)} (n3);
	\draw[->,bend right] (n0) to node[swap] {\(\id_X\)} (n2);
	\draw[->,bend left] (n1) to node[] {\(\id_Y\)} (n3);
	
	\node[rotate=45] at ($(n1)!.5!(n3)!.2!(n2)$)  {\(\leq\)};
\end{tikzpicture}\]
The components $e$ and $p$ are called \emph{embedding} (of $X$ in $Y$) and \emph{projection} of ($Y$ in $X$), respectively, and uniquely determine each other. Since complete partial orders are small categories, embedding-projection pairs are \emph{coreflections};\footnote{A coreflection is an adjunction whose unit is a pseudo-cell.} henceforth we use the two terms interchangeably in the context of \Cpo-categories.

Coreflections in a \Cpo-category \cat{C} form a sub-\Cpo-category of \cat{C} whose objects are those of \cat{C} and whose arrows are embedding-projections with the order on hom-sets given by the ordering on the embeddings (note that $e \leq e' \iff p \geq p'$). We shall write $\cat{C}^\lhd$ for such category.  By forgetting either the projection or embedding  part of a coreflection we get the categories $\cat{C}^e$ and $\cat{C}^p$ (of embeddings and projections), respectively, and such that $\cat{C}^e \cong \cat{C}^\lhd \cong (\cat{C}^p)\op$.

\begin{proposition}[Limit-colimit coincidence]
	Assume \cat{C} enriched over the category \Cpo. For an $\omega$-chain of coreflections $(e_n \lhd p_n\colon X_n \to X_{n+1})_{n < \omega}$ and a cone of coreflections $(f_n \lhd q_n\colon X \to X_{n})_{n < \omega}$ for it the following are equivalent:
	\begin{enumerate}
		\item 
			the cocone $(f_n\colon X \to X_{n})_{n < \omega}$ is a colimit for
			the embeddings chain $(e_n\colon X_n \to X_{n+1})_{n < \omega}$.
		\item
			the cone $(q_n\colon X_n \to X)_{n < \omega}$ is a limit for
			the projections chain $(p_n\colon X_{n+1} \to X_{n})_{n < \omega}$.
	\end{enumerate}
\end{proposition}

The above is a slight reformulation of the limit-colimit coincidence result used in \cite{ps:siam1982} to solve recursive domain equations with unknowns occurring in covariant and contravariant positions like the well-known domain equation:
\[
	D \cong (D \to D) + At
	\text{.}
\]

In \cite{freyd:ct1991} Peter J.~Freyd introduced the concept of \emph{algebraically compact} categories as an abstract context where to address mixed variance. These categories are characterised by a limit-colimit coincidence property for initial/final sequences of endofunctors in a given class \eg \Cpo-endofunctors as in the definition below.
\begin{definition}
	A \Cpo-category is \emph{\Cpo-algebraically complete} \cite{freyd:ct1991} whenever every \Cpo-endofunctor on it has an initial algebra. Dually, a \Cpo-category is \Cpo-coalgebraically cocomplete whenever every \Cpo-endofunctor on it has a final coalgebra. A \Cpo-category (resp.~\Cpob-category) is \emph{\Cpo-algebraically compact} \cite{freyd:ct1991} if every \Cpo-functor has an initial algebra and a 	final coalgebra and they are canonically isomorphic.
\end{definition}

\begin{proposition}
	\label{thm:cpob-cpo-compactness}
	The following hold true:
	\begin{enumerate}
		\item
			\cite{fiore:domainbook}
			A \Cpo-category with an embedding-initial object and colimits of $\omega$-chains of embeddings is \Cpo-algebraically complete.
		\item 
			\cite{freyd:lmms1992}
			A \Cpo-algebraically complete \Cpob-enriched category is \Cpo-algebraically compact.			
	\end{enumerate}
\end{proposition}
The class of \Cpo-algebraically compact categories is closed under products and duals.
\begin{corollary}
	\label{cor:compact-product-dual}
	Assume \cat{C} and \cat{D} \Cpo-algebraically compact,
	$\cat{C}\op$ and $\cat{C}\times\cat{D}$ are \Cpo-algebraically compact.
\end{corollary}
In particular, if \cat{C} is \Cpo-algebraically compact then so is its free involutory category $\cat{C}\op\times\cat{C}$.

\begin{remark}
	Algebraic compactness is at the core of several works on categorical domain theory, especially by Marcelo Fiore \cite{fiore:domainbook} who refined and extended the theory. We restricted ourselves to \Cpo-algebraic compactness in order to simplify the exposition but results presented in this work can be formulated	in the general setting of pseudo-algebraically compact 2-categories \cite{cfw:lics1998}.
\end{remark}

\section{Higher-order operational behaviours via behaviour passing}
\label{sec:hob}

In this Section we present the main contribution of our paper: a coalgebraic characterisation of higher-order behaviours. In \cref{sec:hob-characterisation} we review the process passing approach and discuss its shortcomings when it is applied in the context of coalgebras. This lead us to propose the behaviour passing approach which yield an abstract and syntax-agnostic notion of functors that model higher-order behaviours. We characterise these higher-order behavioural endofunctors as solutions to suitable recursive equations. In \cref{sec:hob-solutions} we study the problem of finding such solutions: the key idea is to obtain them as limits of specific sequences.  Fundamental to this result is limit-colimit coincidence result recalled in \cref{sec:algebraic-compactness} which thus prompts us to work in \Cpo-enriched settings. Solutions obtained in this way define higher-order behavioural functors back in the ``non-enriched'' settings. In \cref{sec:hob-examples} we discuss some illustrative examples.

\subsection{Characterising higher-order behaviours}
\label{sec:hob-characterisation}

We abstract a family of endofunctors modelling value passing behaviours such as
\eqref{eq:eg-firstorder-beh} as a functor:
\[
	\bef{F}\colon \cat{C}\op\times\cat{C} \to \Endo{\cat{C}}
	\text{.}
\]
Herein we shall refer to functors of this type as 
\emph{behaviour families} and say that a behavioural endofunctor
$B$ \emph{belongs} to a family $\bef{F}$ whenever 
$B \cong \bef{F}(V,W)$ for some $V, W\in \cat{C}$.

\begin{remark}
	The above definition could (and will) be generalised
	to allow parameters to range over categories different
	than the category \cat{C} over which endofunctors in
	the family are defined. In the context this work,
	the case of functors of type 
	$\cat{D}\op\times\cat{D} \to \Endo{\cat{C}}$ 
	requires some additional technical attention and therefore
	it will be covered separately in \cref{sec:hob-generalised}
\end{remark}

For the sake of simplicity, let us consider as a running example the functor $\bef{F}(V,W) = \Id^V + W$ over \Set. A system modelled by a coalgebra for $\Id^V + W$ can, at each step, either input an element of $V$ or terminate producing an element of $W$. Fix a set $P$ of processes\footnote{The core ideas of the following reasoning would be unaffected if an algebra of terms were to be used.}. Coalgebras of type $\bef{F}(P,P)$ model systems whose inputs and outputs are processes and the corresponding notion coalgebraic bisimulation is instantiates as follows:
\begin{definition}
	\label{def:bisim-val}
	For a pair of $(\Id^P + P)$-coalgebras, a relation $R$ between their carrier sets is a bisimulation if, and only if, $(x,y) \in R$ implies that:
	\begin{itemize}
		\item
			if $x \xrightarrow{p} x'$ then $y \xrightarrow{p} y'$ and $(x',y') \in R$;
		\item
			if $x \xrightarrow{} p$ then $y \xrightarrow{} p$;
		\item
			the symmetric of the above conditions.
	\end{itemize}
\end{definition}

Endofunctors in a $\bef{F}$ model behaviours passing values from given sets without assuming any additional structure (\eg a semantic equivalence) and consequently what were meant to be processes are indeed plain values in the above definition and, likewise, any coalgebraic construction instantiated on these functors.
To exemplify this issue consider the coalgebra $h\colon P \to P^P + P$ given by $h(p) = \inj_r(p)$ where $\inj_r$ denotes the coproduct right injection. Although $h$ is carried by the set of processes $P$, we do not require it to coincide with the dynamics assumed for the processes in $P$. The coalgebra $h$ describes processes that promptly terminate returning themselves as their output value---a rather limited dynamics---but, since the outputs are distinct elements of $P$, the greatest bisimulation is the identity relation: Assume $R \subseteq P \times P$, if $(x,y) \in R$ then $h(x) = \inj_r(x)$ is equal to $h(y) = \inj_r(y)$ and hence $x = y$. 
Coalgebras of type $\bef{F}(P,P)$ describe systems that can react to inputs regardless of semantics equivalence: they can distinguish processes meant to be behaviourally indistinguishable. For instance, consider a context $C[-]$ such that $C[p] \to p$ for any $p \in P$, if $r \neq s$ then $r \sim s \centernot\implies C[r] \sim C[s]$ since $C[r] \to r$ and $C[s] \to s$.
From these examples it is clear that \cref{def:bisim-val} does not capture the intuitive semantics equivalence, unless the later is the identity on $P$, which is not in many higher-order calculi.

Approaches found in the literature revolve around two main strategies. 
The first is to restrict to the subclass of coalgebras that do not distinguish inputs and outputs that are meant to be semantically equivalent. Unfortunately, this hinders many of the valuable results offered by the coalgebraic approach: for starters, the final coalgebra is not among them. In fact, the final coalgebra for $\Id^C + D$ is the set of all possibly unlimited $C$-branching trees whose leaves are labelled with elements in $D$. We mention \cite{hl:tlca1995} as an example of this approach. 
The second is to provide an \adhoc definition of bisimulation where values are compared using a suitable equivalence relation. This is the idea behind several process-passing calculi such as CHOCS \cite{thomsen:acta1993} or HO$\pi$ \cite{sangiorgi:ic1996} and--albeit mixed with other techniques--constitute the core of environmental bisimulation \cite{sks:lics2007}. Guided by these examples, \cref{def:bisim-val} can be adapted to consider values up-to some given equivalence $\approx$ as follows:
\begin{definition}
	\label{def:bisim-val-dimmed}
	Let ${\approx} \subseteq P \times P$ be an equivalence relation. A relation $R$ between the carrier sets of two $\bef{F}(P,P)$-coalgebras is a $\approx$-bisimulation if, and only if, $(x,y) \in R$ implies that:
	\begin{itemize}
	\item
		if $x \xrightarrow{p} x'$ and $p \approx q$
		then
		$y \xrightarrow{q} y'$ and $(x',y') \in R$
	\item
		if $x \xrightarrow{} p$ then
		$y \xrightarrow{} q$ and $p \approx q$;
	\item
		the symmetric of the above conditions.
	\end{itemize}
\end{definition}

We remark that \cref{def:bisim-val-dimmed} applies to any equivalence relation over $P$ regardless of its coherence with respect to process semantics. To this end, we need to assume $\approx$ to be the greatest $\approx$-bisimulation for processes under the chosen dynamics. Reworded, the definition of processes, their dynamics, and bisimulation depend on each other.

Assuming a suitable $\approx$ exists, $\approx$-bisimulation does not coincide with coalgebraic bisimulation for $\bef{F}(P,P)$-coalgebras unless we restrict to those that can be seen as $\bef{F}(P/{\approx},P/{\approx})$-coalgebras \ie systems that ``handle bisimilar processes in bisimilar ways''.
\begin{lemma}
	\label{thm:bisim-val-dimmed}
	A relation $R$ is a $\approx$-bisimulation if, and only if,
	it is an $\bef{F}(P/{\approx},P/{\approx})$-bisimulation.
\end{lemma}
\begin{proof}[Proof (sketch)]
	For brevity we consider bisimulations on single systems.
	Let $R$ be a $\approx$-bisimulation for some coalgebra $h$ and define $\overline{h}\colon X/R \to \bef{F}(P/{\approx},P/{\approx})(X/R)$ as:
	\[
		\overline{h}([x]_{R}) \defeq \begin{cases}
			\lambda [p]_{\approx}.[\phi(p)]_{R} & 
				\text{if } g(x) = \inj_l(\phi) \\
			[p]_{\approx} & \text{if } g(x) = \inj_r(p)
		\end{cases}
		\text{.}
	\]
	The coalgebra $\overline{h}$ is well-defined for $R$ is	an $\approx$-bisimulation.
	The implication in the other direction is trivial since $\approx$ is
	an equivalence and \Set has regular epimorphisms.
\end{proof}

\begin{remark}
	Intuitively, dimmed bisimulation \cite{dit:forte2014} compares labelled actions up-to some given equivalence relation that ``dims the light used to distinguish labels'' in order to simplify the automated verification of some properties of stochastic systems. \Cref{def:bisim-val-dimmed} can be seen an extension of dimmed bisimulation to systems with input and output actions. In fact, \cref{thm:bisim-val-dimmed} and its proof readily translate to the setting of \cite{dit:forte2014}.
\end{remark}

\Cref{thm:bisim-val-dimmed} suggests that in order to use coalgebraic bisimulation we have to consider classes of processes with the same semantics instead of plain processes. Further evidence supporting this intuition comes from the fact that bisimulation notions for higher-order calculi found in the literature are usually based on some sort of semantically induced quotient of the process set (\cf \cite{thomsen:acta1993,sangiorgi:ic1996,sks:lics2007,kls:mfps2011}) and, \mutatismutandis, this holds also for approaches based on normal forms.
Nonetheless, a subtle issue remains unresolved: $P$ may not be expressive enough to describe \emph{all possible behaviours}. In other words, there might be systems whose behaviour is modelled by some $\bef{F}(P,P)$-coalgebra but cannot be used as a value whereas higher-order systems are though as operating on systems of the same kind.
Under the light of these observations we propose the use of abstract behaviours instead of processes; formally:
\begin{definition}
	\label{def:hob}
	For $\bef{F}\colon \cat{C}\op \times\cat{C} \to \Endo{\cat{C}}$, an endofunctor $B$ in $\bef{F}$ is said to characterise higher-order behaviours if and only if:
	\begin{equation}
		\label{eq:hob}
		B \cong \bef{F}(|\nu B|,|\nu B|)
		\text{.}
	\end{equation}
\end{definition}

\Cref{def:hob} embraces the inherent circularity of higher-order capturing this defining property in the behavioural functor itself instead of imposing restrictions on the systems considered or requiring \adhoc notions of bisimulation. This approach offers the following key advantages:
\begin{enumerate}
	\item Values are canonically defined in a way that is independent from any syntactic representation of behaviours whereas in the process passing approach it is not as clear how processes and their dynamics are given.
	\item The semantic equivalence on values (\cf $\approx$) is bisimilarity and, by strong extensionality of final coalgebras, it coincides with the identity relation which in turn allows us to apply standard coalgebraic results, like coinductive proof methods.
	\item All behaviours are represented, by definition of final semantics.
\end{enumerate}
Note that in general there are no guarantees about the number of solutions to \eqref{eq:hob}: there may be exactly one, more than one, or even none. For instance, all endofunctors considered in the above examples admit final coalgebras but none satisfies \eqref{eq:hob}.

\subsection{Existence and construction}
\label{sec:hob-solutions}

In this subsection we propose a categorical construction for finding solutions to the equation \eqref{eq:hob} and hence for finding endofunctors modelling higher-order behaviours. In particular, solutions are obtained as fixed points of endofunctors over algebraically compact categories where compactness is due  to the unknown in \eqref{eq:hob} occurring both in covariant and contravariant positions as well. 
Henceforth, we assume $\bef{F}\colon \cat{C}\op \times\cat{C} \to \Endo{\cat{C}}$ to be \Cpo-enriched.

The first step towards rephrasing \eqref{eq:hob} in the language of \Cpo-functors is to define a \Cpo-functor $|\nu-|$ assigning endofunctors over a \Cpo-coalgebraically cocomplete category to the object carrying their final coalgebra.
\begin{lemma}
	\label{thm:cpo-final}
	Assume \cat{C} \Cpo-coalgebraically cocomplete. Any family of assignments $\{G \mapsto |\nu G|\}_{G\in \Endo{\cat{C}}}$ extends to a \Cpo-enriched functor $|\nu-|\colon \Endo{\cat{C}} \to \cat{C}$.
\end{lemma}
\begin{proof}
	On objects, the functor is defined by the given assignment. Each transformation $\phi \in \Endo{\cat{C}}(F,G)$ defines a $G$-coalgebra on the
	carrier of the final $F$-coalgebra:
	\[
		|\nu F| \xrightarrow{\nu F} F|\nu F| \xrightarrow{\phi_{|\nu F|}} G|\nu F|
	\]
	whose coinductive extension to the final $G$-coalgebra (which exists by \Cpo-coalgebraic cocompleteness) defines the action $|\nu \phi|$. It is easy to see that the assignment is functorial and continuous, (by naturality and enrichment, respectively) completing the proof.
\end{proof}
For exposition sake, let us assume a chosen assignment and thus fix $|\nu-|$. We remark that this mild assumption can be avoided by carrying out the constructions below in the (more technically involved) pseudo setting.

For a \Cpo-functor $\bef{F}\colon \cat{C}\op\times\cat{C} \to \Endo{\cat{C}}$,
let $\bef{F}_\nu$ denote:
\[
	\Endo{\cat{C}}\op \times \Endo{\cat{C}} 
	\xrightarrow{|\nu-|\op \times |\nu-|\,}
	\cat{C}\op \times \cat{C} 
	\xrightarrow{\bef{F}\,}
	\Endo{C}
\]
which clearly corresponds to the right hand of \eqref{eq:hob}.
From $\bef{F}_\nu$ define the symmetric endofunctor $\breve{\bef{F}_\nu}$:
\[
	\Endo{\cat{C}}\op \times \Endo{\cat{C}} 
	\xrightarrow{\langle \bef{F}_\nu\op\circ \gamma, \bef{F}_\nu\rangle}
	\Endo{\cat{C}}\op \times \Endo{\cat{C}} 
\]
where $\gamma$ denotes the involution for $\Endo{\cat{C}}\op \times \Endo{\cat{C}}$ \ie the isomorphism $\Endo{\cat{C}}\op \times \Endo{\cat{C}} \cong \Endo{\cat{C}} \times \Endo{\cat{C}}\op$--in the following we shall omit $\gamma$ when clear from the context. Akin to the correspondence between $\bef{F}_\nu$ and \eqref{eq:hob}, $\breve{\bef{F}_\nu}$ corresponds the system:
\[
	\begin{cases}
		B \cong \bef{F}(|\nu D|,|\nu B|)\\
		D \cong \bef{F}\op(|\nu B|,|\nu D|)
	\end{cases}
\]
whose solutions are all $\breve{\bef{F}_\nu}$-invariant objects (\ie fixed points).

Algebras for symmetric endofunctors like $\breve{\bef{F}_\nu}$ are suitable pairs of algebras and coalgebras called dialgebras \cite{glimming:calco2007}. Whenever it exists, the initial algebra of a symmetric endofunctor determines the final coalgebra and \viceversa. In fact, any invariant of a symmetric endofunctor determines a specular one by symmetry (just swap the values for the unknowns $B$ and $D$). An invariant is isomorphic to its specular if, and only if, $B \cong D$ and these are all the solutions to \eqref{eq:hob} as illustrated by its reformulation:
\[
	\begin{cases}
		B \cong \bef{F}(|\nu D|,|\nu B|)\\
		D \cong \bef{F}\op(|\nu B|,|\nu D|)\\
		B \cong D
	\end{cases}
\]
Note that initial and final invariants might not be isomorphic. Algebraic compactness ensures that they are. Moreover, any other invariant is \emph{dominated} by the initial/final one in the sense that it factors the isomorphism between the initial and final invariants; in the \Cpo-enriched setting this yields a coreflection from the dominating to the dominated solution.

It is not known whether the class of \Cpo-algebraically compact categories is closed under exponentiation \cite{fiore:domainbook}; it remains an open question  even if we restrict to the case of ``self-exponentials'' such as $\Endo{\cat{C}}$. We mention from \loccit that algebraic compactness is preserved when the exponential base is \Cpob or, in general, any \emph{algebraically super-compact} category and refer the interested reader to \cite{fiore:domainbook} for further details.

Although we cannot state that if \cat{C} is \Cpo-algebraically compact so is \Endo{\cat{C}}, a weaker result will isuffice for our aims. In fact, unlike arbitrary endofunctors over $\Endo{\cat{C}}\op \times \Endo{\cat{C}}$, $\breve{\bef{F}_\nu}$ factors through the \Cpo-algebraically compact category $\cat{C}\op\times\cat{C}$ as:
\[
	\breve{\bef{F}_\nu} 
	= 
	\langle \bef{F}_\nu\op,\bef{F}_\nu \rangle 
	\cong 
	\langle \bef{F}\op,\bef{F} \rangle \circ (|\nu-|\op\times|\nu-|)
	\text{.}
\]
This observation suggests that we can equivalently look for solutions in $\cat{C}\op\times\cat{C}$---as formalised by \cref{thm:factoring-compact} below.

For a \cat{V}-category \cat{E} let $\EndoFact{\cat{E}}{\cat{D}}$ be the sub-\cat{V}-category of \Endo{\cat{E}} whose objects factor through \cat{D}. The category \cat{E} is at least as algebraically compact with respect to \cat{V}-functors in $\EndoFact{\cat{E}}{\cat{D}}$ as \cat{D} is with respect to \cat{V}-endofunctors in \Endo{\cat{D}}.
\begin{lemma}
	\label{thm:factoring-compact}
	For any \Cpo-category \cat{E} and a \Cpo-endofunctor 
	$F \in \EndoFact{\cat{E}}{\cat{D}}$, the following hold true:
	\begin{enumerate}
		\item 
			If \cat{D} is \Cpo-algebraically complete then, $F$ has an initial algebra.
		\item 
			If \cat{D} is \Cpo-coalgebraically cocomplete then, $F$ has a final coalgebra.
		\item
			If \cat{D} is \Cpo-algebraically compact then, $F$ has an initial algebra, a final coalgebra, and they are canonically isomorphic.
	\end{enumerate}
\end{lemma}
\begin{proof}
	By hypotheses there are ${G\colon \cat{D} \to \cat{E}}$ and ${H\colon \cat{E} \to \cat{D}}$ such that $F = G\circ H$. By \Cpo-algebraic completeness of \cat{D} the \Cpo-endofunctor	$H\circ G$ admits an initial algebra $h\colon HGX \to X$; therefore there is at least one $F$-invariant: $Gh \colon (GH)GX \to GX$. For an $F$-algebra $g\colon FY \to Y$ consider the $HG$-algebra $Hg$ and let $x\colon X \to HY$ be its inductive	extension. The morphism $g\circ Gx\colon GX \to Y$ in \cat{E} defines an $F$-algebra morphism from $Gh$ to $g$ proving that the former is weakly initial.	By initiality of $h\colon HGX \to X$, we have that $Hy = x \circ h$ for any algebra morphism $y\colon GX \to Y$ going from $Gh$ to $g$ and hence that:
	\[
		Fy = GHy = G(x\circ h) = F(g \circ Gx) = G(Hg \circ x)
		\text{.}
	\]
	Since $y$ and $g \circ Gx$ are $F$-algebra morphisms we have:
	\[
		g \circ G(x\circ h) = F(g \circ Gx) = g \circ Fy = y \circ Gh
	\]
	and hence $g \circ Gx\circ Gh$ is unique up-to the isomorphism $Gh$ \ie the chosen initial $F$-algebra. Dually, if \cat{D} is \Cpo-coalgebraically cocomplete then $F$ has a final coalgebra. If \cat{D} is \Cpo-algebraically compact then the canonical isomorphism between the initial $HG$-algebra and the final $HG$-coalgebra yields a canonical isomorphism between the initial $F$-algebra	and the final $F$-coalgebra.
\end{proof}

Intuitively, \cref{thm:factoring-compact} says that solving \eqref{eq:hob} for the endofunctor or for its final coalgebra is essentially the same. In fact, we can define a symmetric endofunctor on $\cat{C}\op\times\cat{C}$ starting from $\bef{F}$ and $|\nu-|$ and such that its initial algebra determines the initial algebra for $\breve{\bef{F}_\nu}$ and \viceversa, namely:
\[(|\nu-|\op\times|\nu-|)\circ \langle \bef{F}\op, \bef{F} \rangle\text{.}\]

We are now able to state the main result of this section, namely the existence
of unique (up-to iso) dominating solutions to \eqref{eq:hob}.

\begin{theorem}
	\label{thm:solutions}
	Assume \cat{C} \Cpo-algebraically compact. For any \Cpo-enriched $\bef{F}\colon \cat{C}\op\times\cat{C} \to \Endo{\cat{C}}$, \eqref{eq:hob} admits a unique (up-to isomorphism) dominating solution.
\end{theorem}

\begin{proof}
	By \cref{cor:compact-product-dual} the category $\cat{C}\op\times\cat{C}$ is \Cpo-algebraically compact. By \cref{thm:cpo-final}, the symmetric endofunctor $\breve{\bef{F}_\nu}$ is a \Cpo-endofunctor. Since $\breve{\bef{F}_\nu}$ factors through a \Cpo-algebraically compact category	it is possible to apply \cref{thm:factoring-compact} and thus the initial $\breve{\bef{F}_\nu}$-algebra and the final $\breve{\bef{F}_\nu}$-coalgebra exists and are canonically isomorphic yielding, by symmetry, the required endofunctor $B$ over \cat{C}. Finally, initiality/finality ensures that other solution necessarily factors the aforementioned canonical isomorphism and hence is dominated by $B$.
\end{proof}

The 2-categorical structure embodied by the order-enrichment is purely functional to achieving algebraic compactness. Indeed, any solution obtained in the settings of \cref{thm:solutions} yields a (ordinary) endofunctor solution to \eqref{eq:hob} where the the family of behavioural endofunctors considered is composed by any endofunctor underlying a \Cpo-endofunctor described by $\bef{F}$.

\begin{corollary}
	\label{cor:solutions-not-enriched}
	Assume \cat{C} \Cpo-algebraically compact. For any \Cpo-enriched $\bef{F}\colon \cat{C}\op\times\cat{C} \to \Endo{\cat{C}}$ there exists a unique (up-to iso) dominating solution to $B \cong \underlying{\underlying{\bef{F}}(|\nu B|,|\nu B|)}$.
\end{corollary}

\begin{figure}[t]
	\centering
	\begin{tikzpicture}[font=\footnotesize,auto,
		baseline=(current bounding box.center),
		xscale=1.6, yscale=1.3,
		proj/.style={->>,transform canvas={xshift=-.6ex}},
		emb/.style={>->,transform canvas={xshift=.6ex}},
		]
				
		\node (n-0-1) at (0,6) {\(1\)};
		\node (n-0-2) at (1,6) {\(1\)};
		\node (n-0-3) at (2,6) {\(1\)};
		\node (n-0-4) at (2.8,6) {\(\dots\)};
		\node (n-0-5) at (3.1,6) {\(1\)};
	
		\node (n-1-1) at (0,5) {\(1\)};
		\node (n-1-2) at (1,5) {\(\bef{F}_{1,1}(1)\)};
		\node (n-1-3) at (2,5) {\(\bef{F}^2_{1,1}(1)\)};
		\node (n-1-4) at (2.8,5) {\(\ \cdots\)};
		\node (n-1-5) at (3.1,5) {\(Z_1\)};
		\node[anchor=west] (n-1-6) at (3.2,5) {\(\cong \bef{F}_{1,1}(Z_1)\)};
		
		\node (n-2-1) at (0,4) {\(1\)};
		\node (n-2-2) at (1,4) {\(\bef{F}_{Z_1,Z_1}(1)\)};
		\node (n-2-3) at (2,4) {\(\bef{F}^2_{Z_1,Z_1}(1)\)};
		\node (n-2-4) at (2.8,4) {\(\ \cdots\)};
		\node (n-2-5) at (3.1,4) {\(Z_2\)};
		\node[anchor=west] (n-2-6) at (3.2,4) {\(\cong \bef{F}_{Z_1,Z_1}(Z_2)\)};

		\node (n-3-1) at (0,3) {\(1\)};
		\node (n-3-2) at (1,3) {\(\bef{F}_{Z_2,Z_2}(1)\)};
		\node (n-3-3) at (2,3) {\(\bef{F}^2_{Z_2,Z_2}(1)\)};
		\node (n-3-4) at (2.8,3) {\(\ \cdots\)};
		\node (n-3-5) at (3.1,3) {\(Z_3\)};
		\node[anchor=west] (n-3-6) at (3.2,3) {\(\cong \bef{F}_{Z_2,Z_2}(Z_3)\)};
				
		\node (n-4-1) at (0,2) {\(\vdots\)};
		\node (n-4-2) at (1,2) {\(\vdots\)};
		\node (n-4-3) at (2,2) {\(\vdots\)};
		\node (n-4-4) at (2.8,2) {\(\ \ddots\)};
		\node (n-4-5) at (3.1,2) {\(\vdots\)};
		\node[anchor=west] (n-4-6) at (3.2,2) {\(\)};
	
		\node (n-5-1) at (0,1.5) {\(1\)};
		\node (n-5-2) at (1,1.5) {\(\bef{F}_{Z,Z}(1)\)};
		\node (n-5-3) at (2,1.5) {\(\bef{F}^2_{Z,Z}(1)\)};
		\node (n-5-4) at (2.8,1.5) {\(\ \cdots\)};
		\node (n-5-5) at (3.1,1.5) {\(Z\)};
		\node[anchor=west] (n-5-6) at (3.2,1.5) {\(\cong \bef{F}_{Z,Z}(Z)\ \ \)};
			
		\draw[<-] (n-0-1) -- (n-0-2);
		\draw[<-] (n-0-2) -- (n-0-3);
		\draw[<-] (n-0-3) -- (n-0-4);
		
		\draw[<-] (n-1-1) -- (n-1-2);
		\draw[<-] (n-1-2) -- (n-1-3);
		\draw[<-] (n-1-3) -- (n-1-4);
		
		\draw[<-] (n-2-1) -- (n-2-2);
		\draw[<-] (n-2-2) -- (n-2-3);
		\draw[<-] (n-2-3) -- (n-2-4);
		
		\draw[<-] (n-3-1) -- (n-3-2);
		\draw[<-] (n-3-2) -- (n-3-3);
		\draw[<-] (n-3-3) -- (n-3-4);
		
		\draw[<-] (n-5-1) -- (n-5-2);
		\draw[<-] (n-5-2) -- (n-5-3);
		\draw[<-] (n-5-3) -- (n-5-4);
		
		\draw[emb] (n-0-1) -- (n-1-1);
		\draw[proj] (n-1-1) -- (n-0-1);
		\draw[emb] (n-0-2) -- (n-1-2);
		\draw[proj] (n-1-2) -- (n-0-2);
		\draw[emb] (n-0-3) -- (n-1-3);
		\draw[proj] (n-1-3) -- (n-0-3);
		\draw[emb] (n-0-5) -- (n-1-5);
		\draw[proj] (n-1-5) -- (n-0-5);
		
		\draw[emb] (n-1-1) -- (n-2-1);
		\draw[proj] (n-2-1) -- (n-1-1);
		\draw[emb] (n-1-2) -- (n-2-2);
		\draw[proj] (n-2-2) -- (n-1-2);
		\draw[emb] (n-1-3) -- (n-2-3);
		\draw[proj] (n-2-3) -- (n-1-3);
		\draw[emb] (n-1-5) -- (n-2-5);
		\draw[proj] (n-2-5) -- (n-1-5);
		
		\draw[emb] (n-2-1) -- (n-3-1);
		\draw[proj] (n-3-1) -- (n-2-1);
		\draw[emb] (n-2-2) -- (n-3-2);
		\draw[proj] (n-3-2) -- (n-2-2);
		\draw[emb] (n-2-3) -- (n-3-3);
		\draw[proj] (n-3-3) -- (n-2-3);
		\draw[emb] (n-2-5) -- (n-3-5);
		\draw[proj] (n-3-5) -- (n-2-5);
		
		\draw[emb] (n-3-1) -- (n-4-1);
		\draw[proj] (n-4-1) -- (n-3-1);
		\draw[emb] (n-3-2) -- (n-4-2);
		\draw[proj] (n-4-2) -- (n-3-2);
		\draw[emb] (n-3-3) -- (n-4-3);
		\draw[proj] (n-4-3) -- (n-3-3);
		\draw[emb] (n-3-5) -- (n-4-5);
		\draw[proj] (n-4-5) -- (n-3-5);
			
	\end{tikzpicture}
	\caption{Computing solutions to \protect\eqref{eq:hob} by means of bi-chains.}
	\label{fig:unfolding-thm-solution}
\end{figure}
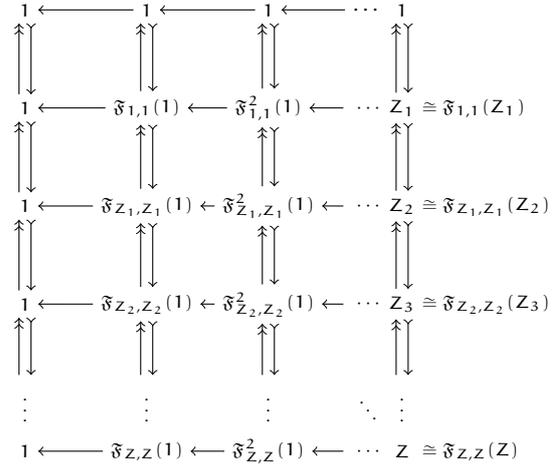

The proof of \cref{thm:solutions} is based on existence of initial and symmetric invariants for $\breve{\bef{F}_\nu}$ and, under the assumption of \cref{thm:cpob-cpo-compactness}, these can be computed via initial sequences for $\breve{\bef{F}_\nu}$. By unfolding the sequence leading to such invariants we obtain the diagram illustrated in \cref{fig:unfolding-thm-solution}. 
Horizontal arrows form final sequences which in turn determine the final coalgebras to be used for instantiating the behavioural functor of the successive iteration. Vertical arrows form chains of embedding-projection pairs and hence characterise horizontal layers as approximations converging to the limiting final sequence depicted in the bottom of the diagram. From this perspective, horizontal layers can be read as \emph{finite-order} behaviours approximating the higher-order solution. In fact, the first row characterises the null-order \eg base values, atoms \etc (recall that the final object in a slice category ${\cat{C} \downarrow V}$ is $\id_V$); behaviours for $\bef{F}(1,1)$ exchange base values and hence are first-order processes and so on. In general, behaviours for $\bef{F}(Z_n,Z_n)$ exchange (abstract) behaviours of order $n$ end hence belong to the order $n+1$. This $\omega$-sequence is limited by higher-order behaviours.

\subsection{Examples}
\label{sec:hob-examples}

\paragraph{Deterministic higher-order behaviours}
In \cref{sec:hob-characterisation} we considered as a running example the parameterised endofunctor $\bef{F}(V,W) = \Id^V + W$ over \Set. Since the cardinality of the set carrying the final $\bef{F}(V,W)$-coalgebra always exceeds both $|V|$ and $|W|$ there are no solutions to \eqref{eq:hob}. Behaviours characterised by this functor are closely related to the domain equation $D \cong (D \to D) + D$: this equation cannot be solved in \Set but admits a unique dominating solution in \Cpob. This observation prompted us to study $\Id^V + W$ as an endofunctor over \Cpob.

Let $(X \to_\perp Y)$ denote the space of continuous bottom-strict functions ordered pointwise and consider $\bef{F}(V,W) = (V \to_\perp \Id) + W$. For any $V$ and $W$, the final $\bef{F}(V,W)$-coalgebra exists and determines all trees whose leaves are in $W$ and whose branches are indexed (on each node) by continuous bottom-strict functions from $V$. Intuitively, bottom elements can be read as unresponsive behaviours like deadlocks and hence inputs modelled by $(V \to_\perp \Id)$ force a system to deadlock whenever it inputs $\perp_V$. Behaviours that deadlock on any input or terminate with output $\perp$ coincide for coproducts are strict. Thus, $\bef{F}(V,W)$ captures eager deterministic computations.

Since the functor $\bef{F}$ is \Cpo-enriched and \Cpob is \Cpo-algebraically compact we conclude by \cref{thm:solutions} that there exists an endofunctor in $\bef{F}$ solution to \cref{def:hob} and that it can be computed as the limit of the sequence depicted in \cref{fig:unfolding-thm-solution}. The fixed point is reached after the first iteration for $1 \cong |\bef{F}(1,1)|$. Indeed, $1$ solves $D \cong (D \to_\perp D) + D$ in \Cpob.

Inspired by the intuitive correlation with domain equations, consider $\bef{F}(V,W) = (V \to_\perp \Id) + W + A$ where $A \neq 1$. Computations are strict as before but now behaviours can terminate returning an atom from $A$. \Cref{thm:solutions} applies also to this family of endofunctors but (thanks to $A \neq 1$) the sequence leading to $B$ contains countably many iterations. Intuitively, abstract behaviours described by the final $B$-coalgebra of are trees with atoms and abstract behaviours as leaves and children indexed by abstract behaviours provided such indexing respect continuity and strictness.

The functors described above are examples of polynomial functors parameterised by $(V, W)$ like those generated by the following simple grammar:
\begin{align*}
	\bef{F}_i(V,W) \Coloneqq\ &
	(X \to_{\perp} \Id) \mid X \mid (V \to_{\perp} \Id) \mid W \mid \\ &
	\bef{F}_1(V,W) + \bef{F}_2(V,W) \mid \\ &
	 \bef{F}_1(V,W) \times \bef{F}_2(V,W)
\end{align*}
All these functors meet the hypotheses of \cref{thm:solutions}.

\paragraph{Non-deterministic higher-order behaviours}
Bounded non-determinism is modelled in the context of \Set by means of the bounded powerset functor $\mathcal{P}_{\!\omega}$. Families of behavioural endofunctors of practical interest and based on $\mathcal{P}_{\!\omega}$ hardly contain endofunctors modelling higher-order behaviours. For instance, \eqref{eq:hob} does not admit any solution when instantiated to \eqref{eq:eg-parametric-beh}. Thus, we model non-deterministic higher-order computations in \Cpob.

Let $\mathbb{B}$ be the boolean lattice. For a map $\phi\colon X \to \mathbb{B}$ in \Cpob, the subset $\phi^{-1}(\top)$ of $X$ is upward closed and does not contain $\perp$ ($\phi$ is monotone and bottom-strict); these subsets of $X$ are equivalent to $\Cpob(X,\mathbb{B})$. Likewise, $\Cpo(X,\mathbb{B})$ determines all upward closed subsets of $X$ since a map $\phi\colon X \to \mathbb{B}$ in \Cpo may map $\bottom_X$ to $\top$. If the order on $X \in \Cpo$ is the anti-chain ordering then any subset of $X$ is trivially upward closed. In fact, the endofunctor $(\Id \to \mathbb{B})$ over \Cpo yields $\mathcal{P}$ by composition with the forgetful functor $U\colon \Cpo \to \Set$ and with the insertion functor $I\colon \Set \to \Cpo$ (which equips each set with the anti-chain ordering). Finally, note that $(\Id \to \mathbb{B})\colon \Cpo \to \Cpo$ restricts to \Cpob. Thus, $(\Id \to \mathbb{B})$ is a good candidate for modelling non-determinism in the contexts of \Cpo and \Cpob.

Both $(\Id \to_{\perp} \mathbb{B})$ and $(\Id \to \mathbb{B})$ are \Cpo-endofunctors over \Cpob; their final coalgebras are carried by $1$ and $\omega$, respectively. This difference reflects the kind of non-deterministic behaviours captured by the two endofunctors. In the first case, a coalgebra can either map a state to the empty set or to some upset that does not contain the bottom element for these are described by strict functions to $\mathbb{B}$---reworded, behaviours are either stuck or able to proceed in a non-deterministic fashion. In the second case, a coalgebra can map a state to any upset meaning that behaviours may get stuck at any time, non-deterministically.

\paragraph{Higher-order CCS}
The late semantics of the CCS with values has been shown in \cite{ft:lics2001} to be captured by the (parameterised) endofunctor \eqref{eq:eg-firstorder-beh} over \Set. However, there is no set of values such that the resulting endofunctor meets the condition in \cref{def:hob}.
Akin to the previous examples, we move from \Set to \Cpob in order to define endofunctors modelling non-deterministic processes that synchronously exchange values along the lines of \eqref{eq:eg-firstorder-beh}.

Fix two objects $V$ and $C$ for values and channels, respectively. Deterministic outputs over channels are characterised by the endofunctor 
$C \times V \times \Id$. Deterministic inputs are described by the endofunctor $C \times (V \to \Id)$. Note that the function space includes also non-strict functions meaning that receiving $\perp$ does not force a system to deadlock---we are interested in the late interpretation of value passing. The non-deterministic component of the behaviour is provided by the ``strict upset'' endofunctor $\mathcal{U}_{\!\perp} = (\Id \to_\perp \mathbb{B})$ since this choice ensures that a process is either stuck or can non-deterministically perform an output, an input, or a silent action. By combining these elements we obtain the family of endofunctors:
\[
	\bef{F}(V,W) = \mathcal{U}_{\!\perp}(
		\overbrace{C \times W \times \Id}^{\text{output}} + 
		\overbrace{C \times (V \to \Id)}^{\text{input}} + 
		\overbrace{\Id}^{\tau}
	)\text{.}
\]
By construction, $\bef{F}$ is \Cpo-enriched and, by \cref{thm:solutions}, there is an endofunctor modelling higher-order systems.

\section{Lifted and dropped solutions}
\label{sec:hob-generalised}

By inspecting the diagram shown in \cref{fig:unfolding-thm-solution} it is clear that limit-colimit coincidence is required only to ensure sequences of embedding-projection pairs are limited. Reworded, only the category of parameters is required to be algebraically compact. In \cref{sec:hob} parameters and behaviours are modelled in the same category but this observation suggest the possibility of considering functors of type $\cat{D}\op\times\cat{D} \to \Endo{\cat{C}}$ where only $\cat{D}$ is assumed algebraically compact. Although this result may appear mainly technical, it allows us to cover a wider class of behaviours since it may often be useful, if not outright necessary, to model parameters and behaviours in different categories in order to simplify the computations of \cref{thm:solutions} or to cover behaviours not expressible as functors of type $\cat{C}\op\times\cat{C} \to \Endo{\cat{C}}$. For instance, behavioural endofunctors might be defined on a (suitably enriched) category of spaces whereas parameters are restricted to range over its subcategory of exponentiable ones. Likewise, one might consider the Kleisli category for a monad and its underlying category---along the lines of \cite{simpson:rr1992}.

The first challenge in characterise endofunctors in $\bef{F}\colon \cat{D}\op\times\cat{D} \to \Endo{\cat{C}}$ that model higher-order behaviours is that the categories where we model systems and exchanged behaviours are distinct. Since higher-order behaviours are meant to operate on behaviours of the same type, we need a way to mediate between their representations in $\cat{D}$ and $\cat{C}$. Intuitively, this means that although systems are modelled as coalgebras for endofunctors over $\cat{C}$, their abstract behaviours (\ie semantics) ``live'' in $\cat{D}$. To this end, we consider behavioural endofunctors with a ``counterpart'' over the category of parameters and some functor between the involved categories that ``mediates'' their behaviours.

\subsection{Families of lifted endofunctors}
\label{sec:hob-lifted}

Here we assume to be given a functor $R\colon \cat{D} \to \cat{C}$ to act as  mediator for behaviours and behavioural endofunctors---the exact meaning of this intuition will be formalised shortly.

An endofunctor $G\in \Endo{\cat{C}}$ is said to be \emph{lifted} along $R$ if there is an endofunctor $H\in \Endo{\cat{D}}$ such that $R \circ H \cong G \circ R$. In order to generalise this condition beyond objects in $\Endo{\cat{C}}$ and $\Endo{\cat{D}}$ consider the following 2-pullback: 
\[\begin{tikzpicture}[
	auto, font=\footnotesize,
	baseline=(current bounding box.center)]
	\begin{scope}[xscale=1.8,yscale=1.1]
		\node (n0) at (0,1) {\(\cat{P}\)};
		\node (n1) at (0,0) {\(\Endo{\cat{D}}\)};
		\node (n2) at (1,1) {\(\Endo{\cat{C}}\)};
		\node (n3) at (1,0) {\(\FCat{\cat{D}}{\cat{C}}\)};
		
		\draw[->] (n0) to node[swap] {\(p_1\)} (n1);
		\draw[->] (n0) to node {\(p_2\)} (n2);
		\draw[->] (n1) to node[swap] {\((R\circ-)\)} (n3); 
		\draw[->] (n2) to node {\((-\circ R)\)} (n3);
		
		\node at (.5,.5) {\(\cong\)};
	\end{scope}
	\draw (n0) ++(7pt,-2pt) -- ++(0pt,-5pt) -- ++(-5pt,0pt);
\end{tikzpicture}\]
Then, the projection of $\cat{P}$ on $\Endo{\cat{C}}$ defines the category of endofunctors lifted along $R$. Formally, we define \lifted{R} as the replete image\footnote{A subcategory \cat{D} of \cat{C} is \emph{replete} if for any $f \in \cat{D}$ and $f\cong g$ in the arrow category $\cat{C}^\to$, then $g \in \cat{D}$ or, equivalently, if the inclusion $\cat{D} \to \cat{C}$ is an isofibration The replete image of a functor $F\colon \cat{C} \to \cat{D}$ is the repletion the image of $F$.} of $p_2$: 
\[
	\lifted{R} \cong \underline{\img}(p_2)
	\text{.}
\]
This definition extends to the \cat{V}-enriched setting as it is.
In particular, for $R$ a \cat{V}-functor, \lifted{R} is the sub-$\cat{V}$-category of $\Endo{\cat{C}}$ formed by all
\begin{itemize}
	\item
		$G$ s.t. for some $H \in \Endo{\cat{D}}$, $R \circ H \cong G \circ R$ in $\FCat{\cat{D}}{\cat{C}}$;
	\item
		$g\colon G \to G'$ s.t., for some $h\colon H \to H'$ in $\Endo{\cat{D}}$, $Rh \cong gR$ in $\FCat{\cat{D}}{\cat{C}}$	
		 (\ie $\psi\circ Rh\circ\phi = g_R$ for $\phi\colon R\circ H \cong G\circ R$ and $\psi\colon R\circ H' \cong G\circ R'$);
	\item
		$g \leq g'$ s.t., for some $h \leq h' \in \Endo{\cat{C}}$, $Rh \cong gR$ and $Rh' \cong g'R$ .
\end{itemize}

In this situation $R$ plays the r\^ole of a mediator between lifted and dropped endofunctors and hence between the behaviours they model. This suggests the following conservative extension of \eqref{eq:hob} to functors of type $\cat{D}\op\times\cat{D} \to \Endo{\cat{C}}$:
\begin{equation}
	\label{eq:hob-lift-fil}
	\begin{cases}
		G \cong \bef{F}(|\nu H|,|\nu H|)\\
		R \circ H \cong G \circ R
	\end{cases}
\end{equation}
where, although the equation systems presents two unknowns, we are actually interested only in $G$ for its values are the endofunctors modelling the systems under scrutiny. However, \eqref{eq:hob-lift-fil} does not offer any correlation between the final coalgebras of $G$ and $H$ that is strong enough for the aims of this work. In fact, for $G$ lifting of $H$ we have that the image through $R$ of the final $H$-coalgebra (canonically) extends to a $G$-coalgebra but not to the final one as illustrated by the diagram:
\begin{equation}
	\label{eq:coalgebra-lift}
	\begin{tikzpicture}[
		auto,font=\footnotesize,
		xscale=2, yscale=1.1,
		baseline=(current bounding box.center)]
			\node (n0) at (0,1) {\(|\nu G|\)};
			\node (n1) at (2,1) {\(G|\nu G|\)};
			\node (n2) at (0,0) {\(R|\nu H|\)};
			\node (n3) at (1,0) {\(RH|\nu H|\)};
			\node (n4) at (2,0) {\(GR|\nu H|\)};
			
			\draw[->] (n0) to 
			node[above] {\(\nu G\)} (n1);
			\draw[->,dashed] (n2) to
				node[left] {\(\psi\)} (n0);
			\draw[->] (n2) to 
				node[above] {\(\nu H\)} (n3);
			\draw[->] (n3) to  
				node[above] {\(\phi_{|\nu H|}\)} (n4);
			\draw[->] (n4) to 
				node[right] {\(G\psi\)}  (n1);
	\end{tikzpicture}
\end{equation}
where $\phi\colon R\circ H \cong G\circ R$ is the natural isomorphism witnessing $G$ as a lifting of $H$. To this end, we need to assume $\psi\colon R|\nu H| \to |\nu G|$ from above to be an isomorphism or, equivalently, $R|\nu H| \cong |\nu G|$ in $\cat{C}$. In such case we say that final invariants lift along the mediator $R$. Formally:

\begin{definition}
	\label{def:final-invariants-lift}
	For $R \colon \cat{D} \to \cat{C}$ 
	we say that \emph{final invariants lift along $R$}
	whenever 
	\[\begin{tikzpicture}[
		auto, font=\footnotesize,
		baseline=(current bounding box.center)]
		\begin{scope}[rotate=45]
			\begin{scope}[scale=1.2]
			\node (n0) at (0.0, 1.0) {\(\cat{P}\)};
			\node (n1) at (0.0, 0.0) {\(\Endo{\cat{D}}\)};
			\node (n2) at (1.0, 1.0) {\(\Endo{\cat{C}}\)};
			\node (n3) at (1.0, 0.0) {\(\FCat{\cat{D}}{\cat{C}}\)};
			\node (n4) at (1.4,-1.4) {\(\cat{D}\)};
			\node (n5) at (2.4,-0.4) {\(\cat{C}\)};
			
			\draw[->] (n0) -- (n1);
			\draw[->] (n0) -- (n2);
			\draw[->] (n1) to node[pos=1,swap] {\((R\circ-)\)} (n3); 
			\draw[->] (n2) to node[pos=1] {\((-\circ R)\)} (n3);
			
			\draw[->] (n1) to node[swap] {\(|\nu-|\)} (n4);
			\draw[->] (n2) to node {\(|\nu-|\)} (n5);
			\draw[->] (n4) to node[swap] {\(R\)} (n5);
			
			\node at (.5,.5) {\(\cong\)};
			\node at (1.6,-.6) {\(\cong\)};
		\end{scope}
		\draw (n0) ++(7pt,-2pt) -- ++(0pt,-5pt) -- ++(-5pt,0pt);
		\end{scope}
	\end{tikzpicture}\]
\end{definition}

Under such conditions, $R$ mediates all abstract behaviours between lifted and dropped endofunctors thus providing the bridge between the category \cat{C} (where systems are modelled) and the category \cat{D} (where parameters range) required to capture the fact that higher-order behaviours operate on behaviours of the same kind. This is captured by extending \eqref{eq:hob-lift-fil} as follows:
\begin{equation}
	\label{eq:hob-lift}
	\begin{cases}
		G \cong \bef{F}(|\nu H|,|\nu H|) \\
		R \circ H \cong G \circ R \\
		R |\nu H| \cong |\nu G|
	\end{cases}
\end{equation}
Then, \cref{def:hob} generalises to this setting as follows:
\begin{definition}
	\label{def:hob-lift}
	For ${\bef{F}\colon \cat{D}\op \times\cat{D} \to \Endo{\cat{C}}}$, $G\colon \cat{C} \to \cat{C}$ is said to characterise higher-order behaviours in $\bef{F}$ if, and only if, there is $H\colon \cat{D} \to \cat{D}$ such that they are solutions to \eqref{eq:hob-lift}.
\end{definition}

Note that \eqref{eq:hob-lift-fil} coincides with \eqref{eq:hob-lift} Under the assumption that final invariants lift along $R$ meaning that we can move some of the information from the equation system to the hypothesis. This simplification is crucial to our aim of obtaining solutions to \eqref{eq:hob-lift} as invariants of suitable endofunctors along the lines of \cref{sec:hob}. Henceforth, we assume $R$ such that final invariants lifts along it. We remark that this holds for any functor preserving final sequences and, in particular, for any right \Cpo-adjoint:
\begin{lemma}
	\label{thm:final-invariants-lift}
	Final invariants lift along right \Cpo-adjoints.
\end{lemma}
\begin{proof}
	The statement can be proved along the lines of	Peter J.~Freyd's ``reflective subcategory lemma'' \cite{freyd:ct1991}.
\end{proof}

Assume, for the argument sake, that there exists an assignment mapping each $G \in \lifted{R}$ to a chosen $\overline{G} \in \Endo{\cat{D}}$ lifting to $G$. Then, we can reformulate \eqref{eq:hob-lift-fil} as follows:
\begin{equation}
	\label{eq:hob-lift-chosen}
	G \cong \bef{F}(|\nu \overline{G}|,|\nu \overline{G}|)
	\text{.}
\end{equation}
Since both \eqref{eq:hob-lift} and \eqref{eq:hob-lift-fil} impose solutions to satisfy the lifting condition we can restrict, without loss of generality, to functors of type $\cat{D}\op\times\cat{D} \to \lifted{R}$. Assume that above assignment extends to a functor $(\overline{-})\colon \lifted{R} \to \Endo{\cat{D}}$, then we can obtain solutions to \eqref{eq:hob-lift-chosen} as invariants by applying the approach described in \cref{sec:hob-solutions} to the functor:
\begin{equation}
	\label{eq:f-chosen-drop}
	\cat{D}\op\times\cat{D}
	\xrightarrow{\bef{F}}
	\lifted{R}
	\xrightarrow{(\overline{-})}
	\Endo{\cat{D}}
	\text{.}
\end{equation}
In general, there might be no (functorial) assignment $(\overline{-})$ but there might be several as well. In the latter case, we remark that they are all equivalent under the assumption that final invariants lift along the mediating functor. In fact, the following holds:
\[\begin{tikzpicture}[
	auto, font=\footnotesize,
	baseline=(current bounding box.center)]
	\begin{scope}[xscale=1.8,yscale=1.1]
		\node (n0) at (0,1) {\(\lifted{R}\)};
		\node (n1) at (0,0) {\(\Endo{\cat{D}}\)};
		\node (n2) at (1,1) {\(\Endo{\cat{C}}\)};
		\node (n3) at (1,0) {\(\FCat{\cat{C}}{\cat{D}}\)};
		
		\draw[->] (n0) to node {I} (n2);
		\draw[->] (n0) to node[swap] {\((\overline{-})\)} (n1);
		\draw[->] (n2) to node {\((-\circ R)\)} (n3);
		\draw[->] (n1) to node[swap] {\((R\circ -)\)} (n3);
		
		\node at (.5,.5) {\(\cong\)};
	\end{scope}
\end{tikzpicture}
\implies
\begin{tikzpicture}[
	auto, xscale=1.8,yscale=1.1, font=\footnotesize,
	baseline=(current bounding box.center)]
	\node (n0) at (0,1) {\(\lifted{R}\)};
	\node (n1) at (0,0) {\(\Endo{\cat{D}}\)};
	\node (n2) at (1,1) {\(\cat{C}\)};
	\node (n3) at (1,0) {\(\cat{D}\)};
	
	\draw[->] (n0) to node {\(|\nu-|\)} (n2);
	\draw[->] (n0) to node[swap] {\((\overline{-})\)} (n1);
	\draw[->] (n1) to node[swap] {\(|\nu-|\)} (n3);
	\draw[->] (n3) to node[swap] {\(R\)} (n2);
	
	\node at (.5,.5) {\(\cong\)};
\end{tikzpicture}
\]
where the inclusion $I$ is given by construction of $\lifted{R}$.
If $R$ is 2-monic\footnote{A morphism $f$ in a 2-category is said to be 2-monic provided that $f\circ g \cong f \circ h \implies g \cong h$.}, then any endofunctor in $\lifted{R}$ is the lifting of a unique (up-to isomorphism) endofunctor:
\begin{lemma}
	\label{thm:chosen-drop}
	For $R$ 2-monic, the diagram below commutes:
	\[\begin{tikzpicture}[
		auto, font=\footnotesize,
		baseline=(current bounding box.center)]
		\begin{scope}[xscale=1.8,yscale=1.1]
			\node (n0) at (0,1) {\(\lifted{R}\)};
			\node (n1) at (1,1) {\(\Endo{\cat{C}}\)};
			\node (n2) at (0,0) {\(\Endo{\cat{D}}\)};
			\node (n3) at (1,0) {\(\FCat{\cat{D}}{\cat{C}}\)};
			
			\draw[->] (n0) -- (n1);
			\draw[dashed,->] (n0) to node[swap] {\((\overline{-})\)} (n2);
			\draw[->] (n1) to node {\((-\circ R)\)} (n3);
			\draw[->] (n2) to node[swap] {\((R\circ-)\)} (n3);
			
			\node at (.5,.5) {\(\cong\)};
		\end{scope}
	\end{tikzpicture}\]\par
\end{lemma}

\begin{proof}
	The functor $(R\circ -)$ is 2-monic since, by hypothesis, $R$ is so. By construction, $p_2\colon \cat{P} \to \Endo{\cat{C}}$ is 2-monic and thus $\cat{P}$ is $\lifted{R}$. The desired $(\overline{-})$ is $p_1$.
\end{proof}

Therefore, \eqref{eq:hob-lift-fil} can be solved by applying the techniques presented in \cref{sec:hob} to $\overline{\bef{F}}\colon \cat{D}\op \times \cat{D} \to \Endo{\cat{D}}$.
\begin{theorem}
	\label{thm:solutions-lift}
	Assume \cat{D} \Cpo-algebraically compact, $R \colon \cat{D} \to \cat{C}$ 2-monic in $\VCat[\Cpo]$, and final invariants to lift along $R$.
	For any \Cpo-functor $\bef{F}\colon \cat{D}\op\times\cat{D} \to \lifted{R}$ the equation system \eqref{eq:hob-lift-fil} admits a unique dominating solution.
\end{theorem}
\begin{proof}
	Apply \cref{thm:solutions} to $\overline{\bef{F}}\colon \cat{D}\op \times \cat{D} \to \Endo{\cat{D}}$ which exists by \cref{thm:chosen-drop}.
\end{proof}

\paragraph{Example: lazy deterministic behaviours}
Let \cat{C} and \cat{D} be \Cpo and \Cpob, respectively, and let $R\colon \cat{D} \to \cat{C}$ be the inclusion $\Cpob \to \Cpo$ whose left \Cpo-adjoint is the lifting functor $(-)_\perp\colon \Cpo \to \Cpob$.

The endofunctor $((V \to \Id) + W)_\perp$ over \Cpob lifts to \Cpo along $R$ as it is. The outer lifting functor $(-)_\perp$ creates a bottom associated with the system being stuck. The function space functor $(V \to \Id)$ describes lazy inputs since these functions are not strict. Finally, the constant functor $W$ models behaviours terminating producing an output in a lazy fashion since the bottom element of $W$ is distinct from the one provided by $(-)_\perp$. Finally, the assignment $(V,W) \mapsto ((V \to \Id) + W)_\perp$ determines a \Cpo-functor $\bef{F}\colon \Cpob\op\times\Cpob \to \lifted{R}$ allowing us to apply \cref{thm:solutions-lift} and obtain the dominating endofunctor among those in $\bef{F}$ that model lazy deterministic higher-order behaviours.

\subsection{Families of dropped endofunctors}
\label{sec:hob-dropped}

In this section we assume the same settings of \cref{sec:hob-lifted} except for $R$ going in the opposite direction. Because the mediating functor $R$ is reversed, we have to consider dropped endofunctors have and ``symmetrise'' all constructions presented in \cref{sec:hob-lifted}. As above, we define endofunctors modelling higher-order behaviours by a conservative extension of \cref{def:hob} and show that such endofunctors exists and can be computed as limits of sequences akin to those we used so far.

The category $\dropped{R}$ of endofunctors \emph{dropped} along $R\colon\cat{C} \to \cat{D}$ is defined, symmetrically to $\lifted{R}$, as the replete image of the projection $p_1\colon \cat{P}\to \Endo{\cat{C}}$ from the 2-pullback:
\[\begin{tikzpicture}[
	auto, font=\footnotesize,
	baseline=(current bounding box.center)]
	\begin{scope}[xscale=1.8,yscale=1.1]
		\node (n0) at (0,1) {\(\cat{P}\)};
		\node (n1) at (0,0) {\(\Endo{\cat{C}}\)};
		\node (n2) at (1,1) {\(\Endo{\cat{D}}\)};
		\node (n3) at (1,0) {\(\FCat{\cat{C}}{\cat{D}}\)};
		
		\draw[->] (n0) to node[swap] {\(p_1\)} (n1);
		\draw[->] (n0) to node {\(p_2\)} (n2);
		\draw[->] (n1) to node[swap] {\((R\circ-)\)} (n3);
		\draw[->] (n2) to node {\((-\circ R)\)} (n3);
		
		\node at (.5,.5) {\(\cong\)};
	\end{scope}
	\draw (n0) ++(7pt,-2pt) -- ++(0pt,-5pt) -- ++(-5pt,0pt);
\end{tikzpicture}\]
This definition extends to the \cat{V}-enriched setting as it is.
In particular, for $R$ a \cat{V}-functor, $\dropped{R}$ is the sub-$\cat{V}$-category of $\Endo{\cat{C}}$ formed by all
\begin{itemize}
	\item
		$H$ s.t. for some $G \in \Endo{\cat{D}}$, $R \circ H \cong G \circ R$ in $\FCat{\cat{C}}{\cat{D}}$;
	\item
		$h\colon H \to H'$ s.t., for some $g\colon G \to G'$ in $\Endo{\cat{D}}$, $Rh \cong gR$ in $\FCat{\cat{C}}{\cat{D}}$	
		 (\ie $\psi\circ Rh\circ\phi = g_R$ for $\phi\colon R\circ H \cong G\circ R$ and $\psi\colon R\circ H' \cong G\circ R'$);
	\item
		$h \leq h'$ s.t., for some $g \leq g' \in \Endo{\cat{D}}$, $Rh \cong gR$ and $Rh' \cong g'R$ .
\end{itemize}

The functor $R$ has to mediate all abstract behaviours between lifted and dropped endofunctors in the sense that, for $G$ lifting of $H$, the image through $R$ of the final $G$-coalgebra (canonically) extends to the final $G$-coalgebra (see \eqref{eq:coalgebra-lift}). Formally: 
\begin{definition}
	\label{def:final-invariants-drop}
	For $R \colon \cat{C} \to \cat{D}$ 
	we say that \emph{final invariants drop along $R$}
	whenever 
\[\begin{tikzpicture}[
	auto, font=\footnotesize,
	baseline=(current bounding box.center)]
	\begin{scope}[rotate=45]
		\begin{scope}[scale=1.2]
		\node (n0) at (0.0, 1.0) {\(\cat{P}\)};
		\node (n1) at (0.0, 0.0) {\(\Endo{\cat{C}}\)};
		\node (n2) at (1.0, 1.0) {\(\Endo{\cat{D}}\)};
		\node (n3) at (1.0, 0.0) {\(\FCat{\cat{D}}{\cat{C}}\)};
		\node (n4) at (1.4,-1.4) {\(\cat{C}\)};
		\node (n5) at (2.4,-0.4) {\(\cat{D}\)};
		
		\draw[->] (n0) -- (n1);
		\draw[->] (n0) -- (n2);
		\draw[->] (n1) to node[pos=1,swap] {\((R\circ-)\)} (n3); 
		\draw[->] (n2) to node[pos=1] {\((-\circ R)\)} (n3);
		
		\draw[->] (n1) to node[swap] {\(|\nu-|\)} (n4);
		\draw[->] (n2) to node {\(|\nu-|\)} (n5);
		\draw[->] (n4) to node[swap] {\(R\)} (n5);
		
		\node at (.5,.5) {\(\cong\)};
		\node at (1.6,-.6) {\(\cong\)};
	\end{scope}
	\draw (n0) ++(7pt,-2pt) -- ++(0pt,-5pt) -- ++(-5pt,0pt);
	\end{scope}
\end{tikzpicture}\]
\end{definition}

Final invariants drop along functors that preserve final sequences and, in particular, along any right \Cpo-adjoint:
\begin{lemma}
	\label{thm:final-invariants-drop}
	Final invariants drop along right \Cpo-adjoints.
\end{lemma}
\begin{proof}
	The statement can be proved along the lines of	Peter J.~Freyd's ``reflective subcategory lemma'' \cite{freyd:ct1991}.
\end{proof}

Akin to \eqref{eq:hob-lift}, \eqref{eq:hob} generalises to $\bef{F}\colon\cat{D}\op \times \cat{D} \to \Endo{\cat{C}}$ as:
\begin{equation}
	\label{eq:hob-drop}
	\begin{cases}
		H \cong \bef{F}(|\nu G|,|\nu G|)\\
		R \circ H \cong G \circ R\\
		R |\nu H| \cong  |\nu G|
	\end{cases}
\end{equation}
where we are actually interested only in the unknown $H$ for its values are the endofunctors modelling the systems under scrutiny. 
\Cref{def:hob} generalises to this setting as follows:
\begin{definition}
	\label{def:hob-drop}
	For ${\bef{F}\colon \cat{D}\op \times\cat{D} \to \dropped{R}}$, an endofunctor is said to characterise higher-order behaviours in $\bef{F}$ if, and only if, it is solution to \eqref{eq:hob-drop}.	
\end{definition}

When final invariants drop along $R$ and ($\bef{F}$ restricts to $\dropped{R}$), \eqref{eq:hob-drop} can be reformulated in the unknown $H$ alone yielding:
\begin{equation}
	\label{eq:hob-drop-fid}
	H \cong \bef{F}(R|\nu H|,R|\nu H|)
	\text{.}
\end{equation}
This equation determines the functor
\begin{equation}
	\label{eq:hob-drop-fid-functorial}
	\cat{C}\op \times \cat{C} 
	\xrightarrow{R\op \times R\,} 
	\cat{D}\op \times \cat{D}
	\xrightarrow{\bef{F}\,}
	\Endo{\cat{C}} 
\end{equation}
and, although its type is the one considered in \cref{sec:hob}, we cannot apply \cref{thm:solutions} since here we do not assume \cat{C} \Cpo-algebraically compact. Therefore, we need to capture \eqref{eq:hob-drop} in terms of endofunctors over $\cat{D}$.

Akin \cref{sec:hob-lifted}, given a functorial assignment choosing liftings for endofunctors dropped along $R$, we reformulate \eqref{eq:hob-drop} as:
\[
	\begin{cases}
		H \cong \bef{F}(|\nu \overline{H}|,|\nu \overline{H}|)\\
		R |\nu H| \cong  |\nu \overline{H}|
	\end{cases}
\]
and, since final invariants drop along $R$, as:
\[
	G \cong \overline{\bef{F}}(|\nu G|,|\nu G|)
\]
which in turn provides us with the desired formulation of \eqref{eq:hob-drop} in terms of endofunctors over $\cat{D}$. Solutions can be characterised as invariants of a suitable endofunctor by applying the approach described in \cref{sec:hob-solutions} to the functor:
\[
	\cat{D}\op \times \cat{D}
	\xrightarrow{\bef{F}\,}
	\dropped{R}
	\xrightarrow{(\overline{-})\,}
	\Endo{\cat{D}}
	\text{.} 
\]
All considerations from \cref{sec:hob-lifted} about the existence of an assignment $(\overline{-})$ apply to this setting. In particular, under the dual hypothesis of \cref{thm:chosen-drop}, there is a unique functor $(\overline{-})$:
\begin{lemma}
	\label{thm:chosen-lift}
	For $R$ 2-epic, the diagram below commutes:
	\[\begin{tikzpicture}[
		auto, font=\footnotesize,
		baseline=(current bounding box.center)]
		\begin{scope}[xscale=1.8,yscale=1.1]
			\node (n0) at (0,1) {\(\dropped{R}\)};
			\node (n1) at (0,0) {\(\Endo{\cat{C}}\)};
			\node (n2) at (1,1) {\(\Endo{\cat{D}}\)};
			\node (n3) at (1,0) {\(\FCat{\cat{C}}{\cat{D}}\)};
			
			\draw[->] (n0) to (n1);
			\draw[dashed,->] (n0) to node {\((\overline{-})\)} (n2);
			\draw[->] (n1) to node[swap] {\((R\circ-)\)} (n3);
			\draw[->] (n2) to node {\((-\circ R)\)} (n3);
			
			\node at (.5,.5) {\(\cong\)};
		\end{scope}
	\end{tikzpicture}\]
\end{lemma}
\begin{proof}
	The functor $(-\circ R)$ is 2-monic since, by hypothesis, $R$ is 2-epic. By construction, $p_1\colon \cat{P} \to \Endo{\cat{C}}$ is 2-monic and thus $\cat{P}$ is $\dropped{R}$. The desired $(\overline{-})$ is $p_1$.
\end{proof}

The approach discussed so far corresponds to the upper path of the following diagram:
\[\begin{tikzpicture}[
	auto, xscale=1.8,yscale=1.1, font=\footnotesize,
	baseline=(current bounding box.center)]
	\node (m-1-1) at (0,0) {\(\cat{D}\op\times\cat{D}\)};
	\node (m-1-2) at (1,0) {\(\dropped{R}\)};
	\node (m-1-3) at (2,0) {\(\cat{C}\)};
	\node (m-2-2) at (1,1) {\(\Endo{\cat{D}}\)};
	\node (m-2-3) at (2,1) {\(\cat{D}\)};
	
	\draw[->] (m-1-1) to node[swap] {\(\bef{F}\)} (m-1-2);
	\draw[->] (m-1-2) to node[swap] {\(|\nu-|\)} (m-1-3);
	\draw[->] (m-1-2) to node[] {\((\overline{-})\)} (m-2-2);
	\draw[->] (m-1-3) to node[swap] {\(R\)} (m-2-3);
	\draw[->] (m-2-2) to node {\(|\nu-|\)} (m-2-3);

	\node at (1.5,.5) {\(\cong\)};
\end{tikzpicture}\]
Since the diagram commutes (final invariants drop along $R$ by assumption), any solution obtained applying \cref{thm:solutions} to $\overline{\bef{F}}$ can be determined also as an invariant of the symmetric endofunctor over $\cat{D}\op\times\cat{D}$ induced by the lower path. Crucially, the latter does not assume the existence of $(\overline{-})$.

\begin{theorem}
	\label{thm:solutions-drop}
	Assume \cat{D} \Cpo-algebraically compact and final invariants to drop along $R \colon \cat{C} \to \cat{D}$. For a \Cpo-functor $\bef{F}\colon \cat{D}\op\times\cat{D} \to \dropped{R}$, \eqref{eq:hob-drop-fid} admits a unique (up-to iso) dominating solution.
\end{theorem}
\begin{proof}
	Let $\breve{\bef{F}}$ denote the symmetric \Cpo-enriched endofunctor over $\cat{D}\op \times \cat{D}$ induced by $R\circ|\nu-| \circ \bef{F}\colon\cat{D}\op \times \cat{D} \to \cat{D}$. By \cref{cor:compact-product-dual} is $\cat{D}\op \times \cat{D}$ \Cpo-algebraically compact and hence both the initial $\breve{\bef{F}}$-algebra and final $\breve{\bef{F}}$-coalgebra exists and are isomorphic yielding, by symmetry $Z \in \cat{D}$ s.t.~$H = \bef{F}(Z,Z)$ is the desired dominating solution to \eqref{eq:hob-drop-fid}.
\end{proof}

\section{Conclusions and future work}
\label{sec:conclusions}

In this paper we have showed how to define behaviours of higher-order systems as final coalgebras of suitable behavioural functors. To this end, we had to solve an intrinsic circularity between the definition of these functors, and their own final coalgebra. We have provided a general construction for both the functors and the final coalgebra, based on a limit-colimit coincidence argument.  Thus, the final coalgebra of such a functor is the object of all abstract higher-order behaviours. As a direct application of the theory of coalgebras, we are now able to define a canonical higher-order coalgebraic bisimulation and show it to coincide with coalgebraic bisimulation for value-passing behaviours. This allows us to employ techniques developed for the seconds (\eg \cite{hl:tcs1995}) to the higher-order case.

These results allow us to contextualize higher-order systems within the general framework of coalgebraic semantics, thus bridging the gap between the two worlds.  It is interesting future work to investigate the application of results from coalgebras, to higher-order systems.  In particular, we expect to be able to define general \emph{higher-order trace equivalences} and \emph{higher-order weak bisimulations} by applying the constructions given in \cite{hjs:lmcs2007,bp:concur2016,bmp:jlamp2015,mp:arxiv2013weak}.

This work sheds some light on the very nature of higher-order behaviour. In fact, there is no general consensus on a \criterium for deciding when a calculus is ``higher-order''. We think that a calculus which is capable to ``deal'' with its own processes via some encoding should not be considered really higher-order, as much as first-order logic is not considered on a par with higher-order logic. The results presented in this paper suggest that \emph{a calculus has an higher-order behaviour if its behavioural functor is defined in terms of its own final coalgebra}. According to this \criterium, $\lambda$-calculus, CHOCS, HO$\pi$, are higher-order, while CCS, $\pi$-calculus, Ambient calculus are not. Indeed there are encodings of higher-order calculi in first-order ones (see \eg  \cite{sw:pibook} for an encoding from $\lambda$-calculus to $\pi$-calculus) suggesting that their targets ``have higher-order semantics''. Nonetheless, these encodings do not reflect strong bisimulations but only weaker notions that hide auxiliary steps introduced by the encoding.

To our knowledge, this is the first general construction of syntax-independent functors for higher-order behaviours.  Many authors have studied labelled transition systems and bisimulations for higher-order calculi (\eg, ``applicative'', ``environmental'', \etc); see \cite{abramsky:rtfp1990,thomsen:acta1993,sks:lics2007, sangiorgi:ic1996, lpss:lics2008} among others.  However, these approaches are tied to the syntactic presentations of specific languages, with no abstract definition of higher-order behaviours. Instead, our work achieves a clear separation between syntax and semantics, which can thus be covered within the bialgebraic theory  of SOS specifications \cite{klin:tcs2011,tp:lics1997}.

The diagram shown in \cref{fig:unfolding-thm-solution} suggests that horizontal layers characterise behaviours approximating higher-order ones. From this perspective, an higher-order behaviour is uniquely defined by countable family of finite-order approximations. We plan to investigate this construction in the context of $\cat{C}$-valued sheaves over the downward topology on the first transfinite ordinal $\omega$. When $\cat{C}$ is \Set, this category is the well-known topos of trees \cite{bmss:lmcs2012}. Since the topology considered has a well-founded base, we would be able to work on higher-order behaviours by inductive arguments on their finite-order approximations.


\end{document}